\documentclass{llncs}
\usepackage{makeidx,enumerate,graphicx,amsfonts,amssymb}
\usepackage{graphicx,epsfig}
\usepackage[lined]{algorithm2e}
\begin{document}

\title{Partial Degree Bounded Edge Packing Problem with Arbitrary Bounds}   
\author{Pawan Aurora \and Sumit Singh \and Shashank K Mehta}         
\institute{Indian Institute of Technology, Kanpur - 208016, India\\
\email{paurora@iitk.ac.in,ssumit@iitk.ac.in,skmehta@cse.iitk.ac.in}}
\tocauthor{Pawan Aurora (Indian Institute of Technology, Kanpur) }

\maketitle

\begin{abstract}
We study the Partial Degree Bounded Edge Packing (PDBEP) problem introduced 
in \cite{Zhang} by Zhang. They have shown that this problem is NP-Hard even for uniform
degree constraint. They also presented approximation algorithms for the case when 
all the vertices have degree constraint of $1$ and $2$ with approximation 
ratio of $2$ and $32/11$ respectively. In this work we study general degree constraint
case (arbitrary degree constraint for each vertex) and present two 
combinatorial approximation algorithms with approximation factors $4$ and
$2$. We also study integer program based solution and present an iterative
rounding algorithm with approximation factor $3/(1-\epsilon)^2$ for any positive
$\epsilon$. Next we study the same problem with weighted edges. In this case
we present an $O(\log n)$ approximation algorithm. Zhang \cite{Zhang} has given an
exact $O(n^2)$ complexity algorithm for trees in case of uniform degree
constraint. We improve their result by giving $O(n\cdot \log n)$ complexity 
exact algorithm for trees with general degree constraint.
\end{abstract}


\section{Introduction}
The {\em partial degree bounded edge packing problem} (PDBEP) is described
as follows:
Given a graph $G=(V,E)$ and degree-bound function $c:V\rightarrow \mathbb{N}$,
compute a maximum cardinality set $E'\subseteq E$ which satisfies the {\em
degree condition}:
$(d'_u \leq c_u) \vee (d'_v \leq c_v)$ for each $e=(u,v)\in E'$.
Here $d'_x$ denotes the degree of vertex $x$ in the graph $G'=(V,E')$.
Without loss of generality, we will assume that $c_v \leq d_v$ for all 
$v\in V$ where $d_v$ denotes the degree of $v$ in $G$.

In the weighted version of the problem edges are assigned non-negative weights 
and we want to compute a set of edges $E'$ with maximum cumulative weight 
subject to the degree condition described above.

In \cite{Zhang}, the PDBEP problem was motivated by an application in binary 
string representation. It was shown there that the maximum expressible 
independent subset (MEIS) problem on $2$-regular set can be reduced to PDBEP 
problem with uniform constraint $c=2$. The PDBEP problem finds another 
interesting application in resource allocation. Given $n$ types of resources 
and $m$ jobs, each job needs two types of resources. A job $j$, which
requires resources $u$ and $v$, can be accomplished if 
$u$ is not shared by more than $c_u$ jobs or $v$ is not shared by more than $c_v$ jobs.
Interpreting the resources as the vertices of the input 
graph and the jobs as edges, the PDBEP problem is to compute the maximum number 
of jobs that can be accomplished.

\subsection{Related Work}
The decision problem of edge packing when there is a uniform  degree constraint 
of $1$ is a parametric dual of the Dominating Set (DS) problem. The 
parametric dual means that for graph $G=(V, E)$, a $k$ sized dominating set 
implies a $|V| - k$ sized edge packing, and vice versa. The parametric dual 
of DS was studied in \cite{Niem}. Further, the dual was well studied under 
the framework of parameterized complexity by Dehne, Fellows, Fernau, Prieto 
and Rosamond in \cite{DFM}.

Recently Peng Zhang \cite{Zhang} showed that the PDBEP problem with uniform 
degree constraint ($c_v=k$ for all $v$) is NP-hard even for $k=1$ for general graphs. 
They gave approximation algorithms for the PDBEP problem under uniform degree 
constraints of $k=1$ and $k=2$ with approximation factors $2$ and $32/11$ 
respectively. They showed that PDBEP on trees with uniform degree constraint
 can be solved in $O(n^2)$ time.

\subsection{Our Contribution}
We propose three different approximation algorithms for the problem with
general degree constraints (i.e., for arbitrary function $c$). Two of these
algorithms are combinatorial in nature and their approximation ratios are
$4$ and $2$. We then show that the relaxation of the natural integer program
 for this problem has a large integrality gap. Then we propose an 
``approximate'' integer program which is a Lagrangian-like relaxation of the 
original IP, and show that any $\alpha$ approximation of this IP is a
$2\alpha/(1-\epsilon)$ approximation of the PDBEP problem for any $\epsilon>0$.
 We then present a $1.5/(1-\epsilon)$ approximation iterative rounding 
\cite{Jain} algorithm for the new integer program. Although this only leads to a
$3/(1-\epsilon)^2$ factor approximation, we think that this Lagrangian like
relaxation is an important
contribution and hope that this method can be applied to some other 
problems to get significantly better results especially in cases where the 
natural IP has a large integrality gap.

The results detailed above are significantly improved results over the 
$2$ and $32/11$ approximations in \cite{Zhang} which are applicable to 
constant function cases $c=1$ and $c=2$ respectively.

Next we consider the PDBEP problem  with general degree constraint for 
edge-weighted graphs. In this case we present a combinatorial approximation 
algorithm with approximation factor of $2+2\log n$.

Finally we present an exact algorithm for unweighted trees with general
degree constraint function. The time complexity of this algorithm is
$O(n\log n)$. This is an improvement over the $O(n^2)$ algorithm
in \cite{Zhang} which is applicable to only a constant degree constraint 
function.

\section{Approximation Algorithms for the unweighted case}

The optimum solution of a PDBEP problem can be bounded as follows.

\begin{lemma} \label{lem1}
Let $G=(V,E)$ be a graph with degree-bound function $c:V\rightarrow \mathbb{N}$.
Then the optimal solution of PDBEP can have at most $\sum_{v\in V} c_v$ edges. 
\end{lemma}

\begin{proof} Let $E'\subset E$ be a solution of PDBEP. Let $U=\{v\in V| d'_v \leq c_v\}$.
Then from the degree condition we see that $U$ is a vertex cover in the graph $(V,E')$.
Hence $|E'|\leq \sum_{u\in U}c_u \leq \sum_{v\in V}c_v$. $\Box$
\end{proof}

\subsection{Edge Addition based Algorithm}

Consider any maximal solution $Y\subseteq E$, i.e., $Y\cup \{e\}$ is not a 
solution for $e\in E\setminus Y$. Let $d_Y(x)$ denote the degree of a vertex
$x$ in the graph $(V,Y)$. Partition the vertex set into sets: 
$A = \{v| d_Y(v) < c_v\}$, $B = \{v |d_Y(v)=c_v\}$, and $C=\{v| d_Y(v) > c_v\}$.
Observe that every edge of the set $E\setminus Y$ which is incident on a vertex in
$A$, has its other vertex in $B$. Hence for any $a_1,a_2\in A$ the $E\setminus Y$ edges 
incident on $a_1$ are all distinct from those incident on $a_2$.
 Construct another edge set $Z$ containing
any $c_v-d_Y(v)$ edges incident on $v$ for each $v\in A$. Observe that $Z$ also satisfies degree
constraint.
Output the larger of $Y$ and $Z$. See Algorithm \ref{edge-addition-algorithm}. 
We have the following result about the correctness.

\begin{lemma}\label{lem2}
The Algorithm \ref{edge-addition-algorithm} outputs a set which satisfies the
degree constraint.
\end{lemma}

Consider the set $Y\cup Z$. In this set the degree of each vertex
is not less than its degree-bound. Hence the cardinality of the output of the
algorithm is at least $\sum_v c_v/4$. From Lemma \ref{lem1} the approximation
ratio is bounded by $4$.

\begin{theorem}\label{thm1} Algorithm \ref{edge-addition-algorithm} has approximation 
factor $4$. 
\end{theorem}

\begin{algorithm}
\label{edge-addition-algorithm}
\KwData{A connected graph $G=(V,E)$ and a function $c: V \rightarrow \mathbb{N}$ 
such that $c_v\leq d(v)$ for each vertex $v$.}
\KwResult{Approximation for the largest subset of $E$ which satisfies the degree-condition.}
$Y:=\emptyset$\;
\For {$e\in E$} {
	\If {$Y\cup \{e\}$ satisfies the degree-condition}{$Y := Y\cup \{e\}$;}}
Compute $A:= \{v\in V| d_Y(v)<c_v\}$\;
$Z:=\emptyset$\;
\For {$v \in A$}{Select arbitrary $c_v-d_Y(v)$ edges incident on $v$ in $E\setminus Y$
and insert into $Z$;}
\eIf {$|Y|\geq |Z|$} {\KwRet $Y$;} {\KwRet $Z$;}
\caption{Edge Addition Based Algorithm}
\end{algorithm}

\subsection{Edge Deletion based Algorithm}
The second algorithm, Algorithm \ref{edge-deletion-algorithm}, for PDBEP is based on elimination 
of edges from the edge set. Starting with the input edge set $E$, iteratively we delete the 
edges in violation, i.e., in each iteration one edge $(u,v)$ is deleted if
the current degree of $u$ is greater than $c_u$ and that of $v$ is greater than $c_v$.
The surviving edge set
$Y$ is the result of the algorithm. Clearly $Y$ satisfies the degree condition. Also
observe that $d_Y(v)\geq c_v$ for all $v\in V$. Hence $|Y|\geq \sum_v c_v/2$.
From Lemma \ref{lem1}, $|Y|\geq OPT/2$.

\begin{theorem}\label{thm2} Algorithm \ref{edge-deletion-algorithm} has approximation ratio $2$. 
\end{theorem}

\begin{algorithm}
\caption{Edge Deletion Based Algorithm}
\label{edge-deletion-algorithm}
\KwData{A connected graph $G=(V,E)$ and a function $c: V \rightarrow \mathbb{N}$ such that $c_v$ is the degree bound for vertex $v$.}
\KwResult{Approximation for the largest subset of $E$ which satisfies the degree-condition.}
$Y := E$\;
\For{$e=(u,v) \in Y$}{
\If{$d_Y(u)>c_u$ and $d_Y(v)>c_v$ } {$Y \leftarrow Y \setminus \{e\}$;}}
\KwRet{$Y$}\;
\end{algorithm}

\subsection{LP based Algorithm}
In this section we explore a linear programming based approach to design
an approximation algorithm for PDBEP.

\subsubsection{The Integer Program}
Following is the natural IP formulation of the problem:

\begin{eqnarray*}
    \textrm{IP1: maximize }  & & \psi=\sum_{e\in E} y_e\  \\
    \textrm{subject to } & & y_e\leq x_u+x_v\; \forall e\in E\\
                    & & \sum_{e\in\delta(v)} y_e \leq c_vx_v+d_v(1-x_v)\; \forall v\in V\\
                    & & x_v\in\{0,1\}\; \forall v\in V\\
                    & & y_e\in\{0,1\}\; \forall e\in E
  \end{eqnarray*}

The solution computed by the program is $E'=\{e|y_e=1\}$.
The linear programming relaxation of the above integer program will be referred to as LP1.

\begin{lemma}\label{lem3}
The integrality gap of LP1 is $\Omega(n)$ where $n$ is the number of vertices in the graph.
\end{lemma}

\begin{proof}
Consider the following instance of the problem. Let $G$ be a complete graph on $n$ vertices
$\{v_0,v_1,\dots,v_{n-1}\}$ and the degree constraint be $c_v=1 \, \forall v \in V$.  
We now construct a feasible fractional solution of $LP1$ as follows. 
Let $x_v=0.5$ for all $v$ and $y_e =1$ for all $e=(v_i,v_j)$ where $j$ is in the interval $(i-\lfloor n/4 \rfloor 
(\textrm{mod}\; n), i+\lfloor n/4 \rfloor (\textrm{mod}\; n))$. The value of 
the objective function for this solution is at least $(n-1)^2/4$. On the other hand, from Lemma 
\ref{lem1}, the optimal solution for the IP1 cannot be more than $n$. Hence the 
integrality gap is $\Omega(n)$. $\Box$
\end{proof}

High integrality gap necessitates an alternative approach.

\subsubsection{Approximate Integer Program}

We propose an alternative integer program IP2 which is a form of Lagrangian relaxation
of IP1. We will show that its {\em  maximal} solutions are also solutions of IP1 and
any $\alpha$ approximation of IP2 is a $2\alpha/(1-\epsilon)$ approximation of IP1.
A maximal solution of IP2 is a solution in which 
$z_v=\max\{0,\sum_{e\in \delta(v)}y_e - c_v\}$
for all $v$ and deletion or addition of an edge does not improve the objective function value.

\begin{eqnarray*}
    \textrm{IP2: maximize } & &\phi = 2\sum_{e \in E} y_e-(1+\epsilon) \sum_{v \in V}z_v, 
\textrm{ for some } \epsilon >0  \\
    \textrm{subject to } & & \sum_{e \in \delta(v)} y_e\leq c_v + z_v\; \forall v\in V\\
                     & & z_v\in \{0,1,2,\dots\}\; \forall v\in V\\
                     & & y_e\in\{0,1\}\; \forall e\in E
  \end{eqnarray*}

Note that any subset of edges $E'$ is a feasible solution of IP2 if we choose $z_v=
\max\{0,\sum_{e\in \delta(v)}y_e -c_v\}$ for all $v$. Besides these values of $z$ will
give maximum value of the objective function. Hence $z$ values are not required to be specified
in the solutions of IP2.

\begin{lemma} \label{lem4}
Every maximal solution  
of the integer program IP2 is also a feasible solution of PDBEP. 
\label{ch4-iterative-feasible}
\end{lemma}

\begin{proof} 
Consider any maximal solution $E'$ of IP2. In a maximal solution 
$z_v=\max\{0,\sum_{e\in \delta(v)}y_e - c_v\}$ for all $v$. Assume that it is not a feasible solution
of PDBEP. Then there must exist an edge $e=(u,v)\in E'$ such that $z_u\geq 1$ and $z_v \geq 1$.
Define an alternative solution $E''=E'\setminus \{e\}$ and decrement $z_u$ and $z_v$ by $1$ each.
Observe that the objective function of the new solution increases by $2\epsilon$. This contradicts
that $E'$ is a maximal solution. $\Box$
\end{proof}

\begin{lemma} \label{lem5} Any $\alpha$ approximate solution of IP2, which is also maximal, is a
 $2\alpha/(1-\epsilon)$ approximation of PDBEP problem.
\end{lemma}

\begin{proof} 
Let $E'$ be an $\alpha$-approximation maximal solution of IP2 with $m_2=\sum_v z'_v$ 
and $m_1=\sum_ey'_e -m_2 = |E'|-m_2$. Let $E''$ be an optimal solution of PDBEP. Then 
$y''_e = 1$ for $e\in E''$ and $z''_v=\max\{0,\sum_{e\in \delta(v)}y''_e - c_v\}$ is 
a solution of IP2, i.e., $E''$ is also a solution of IP2. Define $n_2=\sum_v z''_v$ 
and $n_1=\sum_ey''_e -n_2 = |E''|-n_2$. We have $\phi(E') = 2m_1+(1-\epsilon)m_2$ and 
$\phi(E'') = 2n_1+(1-\epsilon)n_2$. Let $OPT$ denote the optimal value of the IP2 objective
function. Then $\phi(E'') \leq OPT$ and $OPT/\alpha \leq \phi(E')$. So
$2n_1 + (1-\epsilon)n_2 \leq \alpha(2m_1 + (1-\epsilon)m_2)$. So $(1-\epsilon)|E''|
= (1-\epsilon)(n_1+n_2) \leq 2n_1+(1-\epsilon)n_2 \leq \alpha(2m_1+(1-\epsilon)m_2)
\leq 2\alpha (m_1+m_2) = 2\alpha\cdot |E'|$. $\Box$
\end{proof}


\subsection{Algorithm for IP2}
We propose Algorithm \ref{iterative_algorithm} which approximates the IP2 problem within 
a constant factor of approximation. LP2 is the linear program relaxation of IP2.
Here we assume that an additional constraint is imposed, namely, $\{z_v=0| v\in C\}$
where we require a solution in which every $v\in C$ must necessarily satisfy the degree constraint.
The input to the problem is $(H=(V,E),C)$.
Algorithm starts with $E'=\emptyset$ and builds it up one edge at a time by iterative rounding.
In each iteration we discard at least one edge from further consideration. Hence
it requires at most $|E|$ iterations (actually it requires at most $|V|+1$ iterations, see the remark
below.) In the interest of ease in analysis Algorithm \ref{iterative_algorithm} is presented 
in the recursive format.

\begin{eqnarray*}
    \textrm{LP2: maximize } & &\phi = 2\sum_{e \in E} y_e-(1+\epsilon) \sum_{v \in V}z_v,
\textrm{ for some } \epsilon >0  \\
    \textrm{subject to } & & \sum_{e \in \delta(v)} y_e\leq c_v + z_v\; \forall v\in V\setminus C\\
    & & \sum_{e \in \delta(v)} y_e\leq c_v \; \forall v\in C\\
                     & & z_v\geq 0\; \forall v\in V\\
                     & & y_e\geq 0\; \forall e\in E\\
                     & & -y_e\geq -1\; \forall e\in E
  \end{eqnarray*}

\begin{algorithm}
\caption{Iterative Rounding based Algorithm in Recursive Format}
\label{iterative_algorithm}
\KwData{A connected graph $G=(V,E)$ and a function $c: V \rightarrow \mathbb{N}$}
\KwResult{A solution of PDBEP problem.}
	\For {$v\in V$} {$f_v:=c_v$;}
	$C := \emptyset$\;
	$E' := SolveIP2(G,C,f)$\;
	\KwRet{$E'$}\;
\vspace*{3mm}
{\bf Function}: SolveIP2$(H=(V_H,E_H),C,f)$\\
	\If{$E_H :=\emptyset$}{\KwRet{$\emptyset$};}
	Delete all isolated vertices from $V_H$\;
	$({\bf y},{\bf z}) =$ LPSolver$(H,C)$\;
       \tcc{solve LP2 with degree-bounds $f(x)$ for all $x\in V_H$} 
        \eIf {$\exists e \in E_H$ with $y_e=0$}
	{$H_1 := (V_H,E_H\setminus \{e\})$\;
	$C_1 := C$\;
	$E' := SolveIP2(H_1,C_1,f)$;}
	{From Lemma \ref{lem6} there exists an edge $e:=(u,v)$ with $y_e\geq 1/2$\;
	      From Lemma \ref{lem7} w.l.g. we assume $(f_v>0, z_v=0)$\;
	      $f_v:=f_v-1$\;
	      $C_1 := C \cup \{v\}$\;
	      $f_u := \max \{f_u-1,0\}$\;
	      $H_1 := (V_H, E_H\setminus \{e\})$\;
	      $E' := SolveIP2(H_1,C_1,f)\cup \{e\}$\;
              \tcc{By including $e$ in $E'$ we effectively rounded up $y_e$ to $1$. Hence in case $f_u=0$
                   then implicitly $z_u$ is also raised to ensure that the condition $\sum_{e'\in \delta(u)}y_{e'}
                   \leq f_u+z_u$ continues to hold. We do not explicitly increase $z_u$ value since
                   it is not output as a part of the solution.}
             }
	\KwRet{$E'$}\;
\end{algorithm}

In the following analysis we will focus on two problems: $(H,C)$ of some $i$-th
nested recursive call and $(H_1,C_1)$ of the next call. For simplicity we will refer to
them as the problems of graphs $H$ and $H_1$ respectively.

\begin{lemma} \label{lem6} In a corner solution of LP2 on a non-empty graph 
there is at least one edge $e$ with $y_e=0$ or $y_e\geq 1/2$.
\end{lemma}

\begin{proof}
Assume the contrary that in an extreme point solution of LP2 all $y_e$ are in 
the open interval
$(0,1/2)$. Let us partition the vertices as follows.
Let $n_1$ vertices have $f_v>0$ and $z_v>0$, $n_2$ vertices have
$f_v>0$ and $z_v=0$ and $n_3$ vertices have $f_v=0$ and $z_v>0$. Note that the case of 
$f_v=0$ and $z_v=0$ cannot arise because $y_e>0$ for all $e$. In each case let $n'_i$ 
vertices have the condition
$\sum_{e\in \delta(v)} y_e \leq f_v + z_v$ tight (an equality) and $n''_i$ vertices
have the condition a strict inequality. Let the number of edges be $m$.

The total number of variables is $n_1+n_2+n_3+m$. In $n'_1+n_2'$ cases $\sum_{e\in \delta(v)} y_e
= f_v + z_v$ where $f_v\geq 1$ and each $y_e<0.5$ so there must be at least $3$ edges
incident on such vertices. Since the graph has no isolated vertices, every remaining vertex
has at least one incident edge. Hence $m\geq (3n'_1 + 3n'_2+ n''_1+n''_2+n_3)/2$. So
number of variables is at least $n'_1 + n'_2+ (1.5)(n_1+n_2+n_3)$.

Now we find the number of tight conditions. None of the $y_e$ touch their bounds.
The number of $z_v$ which are equal to zero is $n_2$, and the number of instances when
$\sum_{e\in \delta(v)} y_e = f_v + z_v$ is $n'_1+n'_2+n'_3$. Hence the total
number of conditions which are tight is $n_2+n'_1+n'_2+n'_3$. Since the solution is
an extreme point, the number of tight conditions must not be less than the number of
variables. So $n_2+n'_1+n'_2+n'_3 \geq n'_1 + n'_2 + (1.5)(n_1+n_2+n_3)$.
This implies that $n_1=n_2=n_3=0$, which is absurd since the input graph is not empty.
$\Box$
\end{proof}

\noindent
{\bf Remark}: The program LP2 has $|E|+|V|$ variables and $2|E|+2|V|$ constraints. Hence in the 
first iteration the optimal solution must have at least $|E|-|V|$ tight edge-constraints
(i.e., $y_e=0$ or $y_e=1$.) All these can be processed simultaneously so in the second round at 
most $|V|$ edges will remain in the residual graph. Thus the total number of iterations cannot exceed $|V|+1$.

\begin{lemma}\label{lem7} If $y_e\geq 1/2$ in the solution of LP2 where $e=(u,v)$, 
then $(f_u>0, z_u=0)$ or $(f_v>0, z_v=0)$.
\end{lemma}

\begin{proof} 

Assume that $z_v>0$ and $z_u>0$ in the solution. Let minimum of $z_v$, $z_u$, and $y_e$ 
be $\beta$. Subtracting $\beta$ from these variables results in a feasible 
solution with objective function value greater than the optimum by $2\beta\cdot \epsilon$. This
is absurd. Hence $z_u$ and $z_v$ both cannot be positive.

Next assume that $f_u=0$ and $z_u=0$. Then $y_e$ must be zero, contradicting the fact
that $y_e\geq 1/2$. Similarly $f_v=0$ and $z_v=0$ is also not possible.

Therefore either $(f_u>0, z_u=0)$ or $(f_v>0, z_v=0)$.
$\Box$
\end{proof}

\begin{lemma} \label{lem8} The Algorithm \ref{iterative_algorithm} returns a feasible solution 
of PDBEP.
\end{lemma}

\begin{proof}
The claim is trivially true when the graph is empty. We will use induction.

In the case of $y_e=0$, the solution of $H_1$ is also the solution of $H$.
From induction hypothesis it is feasible for $H_1$ hence it is also feasible for $H$.

Consider the  second case, i.e., $y_e\geq 1/2$. Let $e=(u,v)$. From Lemma \ref{lem7} 
$f_v=a>0$ and $z_v=0$. 
Since $z_{1v}=0$ and $f_{1v}=a-1$, in the solution
of $H_1$ at most $a-1$ edges can be incident on $v$. So there are at most
$a$ edges incident on $v$ and $f_v=a$ in the solution of $H$. Thus $e$ is valid in the solution 
of $H$. Other edges are valid due to induction hypothesis. $\Box$
\end{proof}

Now we analyze the performance of the algorithm.

\begin{lemma}\label {lem9} Algorithm \ref{iterative_algorithm} gives a $1.5/(1-\epsilon)$ approximation of IP2.
\end{lemma}

\begin{proof} Let $c$ denote $1.5/(1-\epsilon)$.
We will denote the optimal LP2 solutions of $H$ and $H_1$
by $F$ and $F_1$ respectively. Similarly $I$ and $I_1$ will denote the 
solutions computed by the algorithm for $H$ and $H_1$ respectively. 
$f_{1*}$ and $z_{1*}$ denote the parameters associated with $H_1$.
We will assume that $z_x = \max\{0,\sum_{e\in \delta(x)}y_e - f_x\}$
for integral solutions to compute their $\phi$-values. 
Again we will
prove the claim by induction. The base case is trivially true. From induction hypothesis
$\phi(F_1)/\phi(I_1) \leq c$ and our goal is to show the same bound holds for
$\phi(F)/\phi(I)$.

In the event of $y_e=0$ in $F$, $\phi(F)=\phi(F_1)$ and $\phi(I)=\phi(I_1)$.
Hence $\phi(F)/\phi(I) = \phi(F_1)/\phi(I_1)$.

In case $y_e=\alpha \geq 1/2$ we will consider two cases: (i) $f_u>0$ and (ii) $f_u=0$ in
$F$. In the first case $I$ differs from $I_1$ in three aspects: $y_e=1$ in $I$, $f_v=f_{1v}+1$
and $f_u=f_{1u}+1$. So $z_v$ and $z_u$ remain unchanged, i.e., $z_v=z_{1v}$ and $z_u=z_{1u}$.
Thus $\phi(I)=\phi(I_1)+2$. In the second case also $y_e$ increases by $1$ and
$z_v$ remains unchanged but $z_u$ increases
by $1$ because in this case $f_u=f_{1u}=0$. Hence $\phi(I) = \phi(I_1)+1-\epsilon$.

From induction hypothesis $\phi(F_1) \leq c\phi(I_1)$. Hence for any solution
$F'_1$ of $H_1$, we have $\phi(F'_1) \leq c\phi(I_1)$. In the remaining part of the
proof we will construct a solution of LP2 for $H_1$ from $F$, the optimal solution of LP2 for
$H$.

Again we will consider the two cases separately. First the case of $f_u>0$.
Set $y_e=0$. If $\sum_{e'\in \delta(v)\setminus \{e\}}y_{e'} \geq 1-\alpha$ then subtract the
values of $y_{e'}$ for $e'\in \delta(v)\setminus \{e\}$ in arbitrary manner so that
the sum $\sum_{e'\in \delta(v)\setminus \{e\}}y_{e'}$ decreases by $1-\alpha$. If
$\sum_{e'\in \delta(v)\setminus \{e\}}y_{e'} < 1-\alpha$, then set $y_{e'}$ to zero for all edges
incident on $v$. Repeat this step for edges incident on $u$.
Retain values of all other variables as in $F$ (in particular, the values
of $z_u$ and $z_v$). Observe that these values constitute a solution of LP2 for $H_1$. Call
this solution $F'_1$. Then $\phi(F'_1) \geq \phi(F) - 2(1 + 1-\alpha) \geq \phi(F)-3$.
We have $\phi(F) \leq \phi(F'_1) +3 \leq c(\phi(I_1)) +3 \leq c(\phi(I)-2)+3 \leq c\phi(I)$.

In the second case $f_u=0$. Once again repeat the step described for edges incident on $v$
and set $y_e$ to zero.
In this case $z_u\geq \alpha$ so subtract $\alpha$ from it. It is easy to see that again the
resulting variable values form an LP2 solution of $H_1$, 
call it $F'_1$. So $\phi(F'_1) = \phi(F) -(2-(1+\epsilon)\alpha)$. So $\phi(F) =
\phi(F'_1) + (2-(1+\epsilon)\alpha) \leq c\phi(I_1) + (2-(1+\epsilon)\alpha)$. Plugging 
$\phi(I) -1 +\epsilon$ for $\phi(I_1)$ and simplifying the expression gives
$\phi(F) \leq c\phi(I)$. This completes the proof. $\Box$
\end{proof}

Combining lemmas \ref{lem5} and \ref{lem9} we have the following result.

\begin{theorem} \label{thm3} Algorithm \ref{iterative_algorithm} approximates PDBEP with approximation factor 
$3/(1-\epsilon)^2$.
\end{theorem}

\section{Approximation Algorithm for the weighted case}

Let $H(v)$ denote the heaviest $c_v$ edges incident on vertex $v$, called {\em heavy set} of 
vertex $v$. Then from a generalization of
Lemma \ref{lem1} the optimum solution of PDBEP in weighted-edge case is bounded by
$\sum_{v\in V}\sum_{e\in H(v)}w(e)$ where $w(e)$ denotes the weight of edge $e$.
We will describe a method to construct upto $1+2\log |V|$ solutions, which cover 
$\cup_{v\in V}H(v)$. Then the heaviest solution gives a $2+2\log |V|$ approximation of the problem.

\subsection{The Algorithm}
Input: A graph $(V,E)$ with non-negative edge-weight function $w()$. Let $|V|=n$.

Step 1: $E_1 = E\setminus \{e=(u,v)\in E|e\notin H(u) \textrm{ and } e\notin H(v)\}$.

Step 2: $T = \{e=(u,v)\in E| e\in H(u) \textrm{ and } e\in H(v)\}$.

Step 3: $E_2 = E_1 \setminus T $. Clearly each edge of $E_2$ is in the heavy set of one of its
end-vertices. Suppose $e=(u,v)\in E$ with $e\notin H(u)$ and $e\in H(v)$. Then we will
think of $e$ as directed from $u$ to $v$.

Step 4: Arbitrarily label the vertices from $0$ to $n-1$. Define sets of edges $A_0,
\dots, A_{k-1}$ and $B_0,\dots, B_{k-1}$, where $k=\log n$, as follows. 
$A_r$ consists of edges $(u,v)$ directed from $u$ to $v$, such that the $r-1$
least significant bits of binary expansion of the labels of $u$ and $v$ are same and $r$-th bit of $u$ 
is zero and the same bit of $v$ is one. $B_r$ is defined similarly except the
$r$-th bit of $u$ is one and that of $v$ is zero.

Step 5: Output that set among the $2\log n + 1$ sets, $T,A_0,\dots,A_{k},B_0,\dots,B_{k}$,
which has maximum cumulative edge weight.

\begin{theorem} \label{thm4}
The algorithm gives a feasible solution with approximation factor $2+2\log n$.
\end{theorem}

\begin{proof}
Set $T$ constitutes a feasible solution since both ends of each edge in it satisfy
the degree constraint. In $A_r$ all arrows are pointed from $u$ with $r$-th bit zero
to $v$ with $r$-th bit one. Hence it is a bipartite graph where all arrows have heads in one set
and the tails in the other. All vertices on the head side satisfy the degree conditions 
because all their incident edges are in their heavy sets. Therefore $A_r$ are feasible 
solutions. Similarly all $B_r$ are also feasible.
It is easy to see that every edge of $E_2$ belongs to $A_r$ or $B_r$ for some $r$.
Hence $T\cup (\cup_r A_r) \cup (\cup_r B_r) = E_1$. 
Observe that $\cup_{v}H(v) = E_1$. Only $T$-edges have both ends in heavy sets. Using the fact that
$OPT \leq \sum_v w(H(v))$, we deduce that $OPT \leq 2w(T) + \sum_r (w(A_r)+w(B_r))$. So the weight of the set
output in step 5 is at least $OPT/(2 + 2\log n)$. $\Box$
\end{proof}

\section{Exact Algorithm}
In this section we give a polynomial time exact algorithm for the unweighted PDBEP problem for the 
special case when the input graph is a tree. We will denote the degree of a vertex $v$ in the input graph
by $d(v)$ and its degree in a solution under consideration by $d'(v)$.

Let $T$ be a rooted tree with root $R$. For any vertex $v$ we denote the subtree rooted at
$v$ by $T(v)$. Consider all feasible solutions of PDBEP of graph $T(v)$ in which degree of $v$
is at most $c_v-1$, call them $H$-solutions. Let $h(v)$ be the number of edges in the largest such solution.
Similarly let $g(v)$ be the optimal $G$-solution in which the degree of $v$ is restricted to be equal to
$c_v$. Lastly $b(v)$ will denote the  optimal $B$-solution which are solutions of $T(v)$ under the 
restriction that degree of $v$ be at least $c_v$ and every
neighbor of $v$ in the solution satisfies the degree condition.  It may be observed that one class of
solutions of $T(v)$ are included in $G$-solutions as well as in $B$-solutions. These are the solutions
in which $d'(v)=c_v$ and every neighbor $u$ of $v$ in the solution has $d'(u)\leq c_u$.
If in any of these cases there are no feasible solutions, then the corresponding optimal value
is assumed to be zero. Hence the optimum solution of PDBEP for $T$ is the maximum of $h(R), g(R)$, 
and $b(R)$ and all three values are zero for a leaf nodes. 

Let $Ch(v)$ denote the set of child-nodes of $v$ in $T(v)$. We partition $Ch(v)$ into  
$H(v)=\{u\in Ch(v)|h(u)\geq \max\{g(u),b(u)\}\}$, 
$G(v)=\{u\in Ch(v)|g(u)>\max\{h(u),b(u)\}\}$, 
$B(v)=Ch(V)\setminus (G(v) \cup H(v))$.
While constructing a $G$-solution of $T(v)$ from the solutions of the children of $v$
we can include the edge $(v,u)$ for any vertex
$u$ in $H(v)\cup B(v)$ along with the optimal solution of $T(u)$ without
disturbing the degree conditions of the edges in this solution.
But we can add edge $(v,u)$ to the solution, for any $u\in G(v)$, only by selecting a $B$-solution
or an $H$-solution of $T(u)$ because if we use a $G$-solution for $T(u)$, then vertex $u$ which 
was earlier satisfying the degree condition, will now have degree $c_u+1$. Similarly while
constructing a $B$-solution of $T(v)$ we can connect $v$ to any number of $H(v)$ vertices and use their
optimal $H$-solutions. But in order to connect $v$ with $u\in B(v)\cup G(v)$ we must use the optimal
$H$-solution of $T(u)$.

If $k=c_v-|H(v)|-|B(v)|>0$, then we define $S'(v)$ to be the set of $k$ members of $G(v)$ with
smallest values of $g(u)-\max\{h(u),b(u)\}$. Otherwise $S'(v) = \emptyset$.
Similarly if $k=c_v-|H(v)|>0$, then we define $S''(v)$ to be the set of $k$ members of $G(v)\cup B(v)$ with
smallest key values where key is $g(u)- h(u)$ for $u\in G(v)$ and $b(u)-h(u)$ for $u\in B(v)$. Otherwise 
$S''(v) = \emptyset$. Now we have following lemma which leads to a simple dynamic program for PDBEP.

\begin{lemma} For any internal vertex $v$ of $T$,\\

(i) $h(v) = \sum_{u\in B(v)} b(u) + \sum_{u\in H(v)}h(u) + \sum_{u\in G(v)}g(u)
 + \min\{c_v-1,|H(v)|+|B(v)|\}$,

If $d(v)=c_v$ and $v\neq R$, then set $b(v)=g(v)=0$ otherwise

(ii) $g(v) = \sum_{u\in B(v)} b(u) + \sum_{u\in H(v)}h(u) + \sum_{u\in G(v)\setminus S'(v)}g(u)+ \\
\sum_{u\in S'(v)} \max\{h(u),b(u)\} + c_v$,

(iii) $b(v) =  \sum_{u\in H(v)\cup S''(v)}h(u) +  \sum_{u\in B(v)\setminus S''(v)} b(u) + \sum_{u\in G(v)
\setminus S''(v)}g(u) + \max\{c_v, |H(v)|\}$,
\end{lemma}

Observe that if $h(u)$ is equal to $b(u)$ or $g(u)$, then $u$ is categorized as an $H(v)$ vertex
and if $b(u)=g(u)>h(u)$, then $u$ is assigned to $B(v)$ set. Hence the last term is maximum in 
each of the cases in the lemma.

The algorithm initializes $h(v),b(v)$, and $g(v)$ to zero for the leaf nodes and computes these
values for the internal nodes bottom up. Finally it outputs the maximum of the three values of the root $R$.
Computations for any internal vertex takes $O(|Ch|\log |Ch|)$ time where $Ch$ is the set of children
of that vertex. Besides the ordering the vertices so that child occurs before the parent (topological sort)
takes $O(n)$ time. Hence the time complexity is $O(n\log n)$.

\section{Future work}
It remains an open question if there exists a constant factor approximation algorithm for the PDBEP problem when the input graph is weighted. The objective function of LP1 can be easily modified to handle the weighted case but due to the large integrality gap this approach remains useless. However, there are cutting-plane methods like Chv{\'a}tal-Gomory cuts \cite{ST} that have been known to improve the integrality gaps for some problems. It would be worthwhile to see if these methods can help reduce the integrality gap of LP1.

There seems no reason to assume that the number of resources necessary for the accomplishment of a job in the resource allocation problem cannot exceed two. Hence a natural generalization of the PDBEP problem to hypergraphs.

Any NP-hard problem is not considered resolved unless and until it has an algorithm with approximation factor that matches the lower bound for that problem. As far as we know there is no known inapproximability result for the PDBEP problem. So that presents another avenue for further research.

\bibliographystyle{plain}
\bibliography{pdbep}

\end{document}